\documentclass{llncs}

\usepackage{latexsym,amssymb,amsmath}

\newbox\tempa
\newbox\tempb
\newdimen\tempc
\def\mud#1{\hfil $\displaystyle{\mathstrut #1}$\hfil}
\def\rig#1{\hfil $\displaystyle{#1}$}
\def\irulehelp#1#2#3{\setbox\tempa=\hbox{$\displaystyle{\mathstrut #2}$}%
                        \setbox\tempb=\vbox{\halign{##\cr
        \mud{#1}\cr
        \noalign{\vskip\the\lineskip}%
        \noalign{\hrule height 0pt}%
        \rig{\vbox to 0pt{\vss\hbox to 0pt{${\; #3}$\hss}\vss}}\cr
        \noalign{\hrule}%
        \noalign{\vskip\the\lineskip}%

        \mud{\copy\tempa}\cr}}%
                      \tempc=\wd\tempb
                      \advance\tempc by \wd\tempa
                      \divide\tempc by 2 }
\def\irule#1#2#3{{\irulehelp{#1}{#2}{#3}%
                     \hbox to \wd\tempa{\hss \box\tempb \hss}}}

\newcommand{\lra}{\longrightarrow}

\newcommand{\fa}{\forall}
\newcommand{\dotfa}{\dot{\fa}}
\newcommand{\dotvee}{\dot{\vee}}
\newcommand{\dotneg}{\dot{\neg}}
\def\nulll{\mbox{\it Null\/}}

\begin{document}     
\title{How can we prove that a proof search method is not an instance of 
another?}
\author{Guillaume Burel\inst{1} and Gilles Dowek\inst{2}}
\date{}
\institute{
Nancy-Universit\'e and LORIA\\
{\tt guillaume.burel@ens-lyon.org,
  http://www.loria.fr/\~{}burel}
\and\'Ecole polytechnique and INRIA\\
LIX, \'Ecole polytechnique,
91128 Palaiseau Cedex, France. \\
{\tt gilles.dowek@polytechnique.edu,
  http://www.lix.polytechnique.fr/\~{}dowek}
}
\maketitle

\thispagestyle{empty}

\begin{abstract}
We introduce a method to prove that a proof search method is not an
instance of another. As an example of application, we show that
Polarized resolution modulo, a method that mixes clause selection
restrictions and literal selection restrictions, is not an instance
of Ordered resolution with selection.
\end{abstract}

\section{Introduction}

An important property of the resolution method
\cite{robinson65resolution} is its refutational completeness, that is
the possibility to derive the empty clause from any unsatisfiable set
of clauses.  However, the search space to derive this clause can be
unnecessarily big. For instance, to derive the empty clause from the
set formed with the clause $P$, the clause $\neg P, Q$ and the clause
$\neg Q$, we can first generate the clause $Q$ from $P$ and $\neg P,
Q$ and then the empty clause from this clause and $\neg
Q$. Alternatively, we can generate the clause $\neg P$ from $\neg P,
Q$ and $\neg Q$ and then the empty clause from this clause and $P$.  These
two derivations are redundant and eliminating such redundancies is a
key issue to design an efficient proof search method.  Several types
of restrictions may be applied to eliminate these redundancies, while
preserving refutational completeness.

First, we may impose some restriction on the choice of clauses,
allowing the resolution rule to be applied to some pairs of clauses
and forbidding it to be applied to others. This type of restriction is
used, for instance, in the {\em Hyper-resolution} method
\cite{robinson65automatic}, in the {\em Set-of-support} method
\cite{wos65sos}, and in the {\em Semantic resolution method}
\cite{slagle67automatic}. For instance, in the Set-of-support method,
we identify a consistent subset of the set of clauses to be refuted,
called {\em the theory},
for instance, the subset formed with the clauses $P$ and $\neg P, Q$.
Then, the resolution rule can be applied to a pair of
clauses if at most one clause is in this set,  but it is forbidden
if both are.

Then, we may impose some restriction on the choice of literals, 
allowing the resolution rule to be applied to some literals of the resolved 
clauses and forbidding it to be applied to others.
This type of restriction is used, for instance,  in {\em Ordered 
resolution with selection}. In this method, a selection function and an order
relation define a set of selected literals in each clause. Then,
the resolution rule may be applied to a pair of clauses if 
the resolved literal are selected in the clauses, but it is forbidden 
otherwise.

It is easy to remark that combining clause selection restrictions
and literal selection restrictions in the theory clauses may jeopardize
completeness, even when the theory is consistent and the selected literals
are defined with respect to an order, as in the Ordered resolution with 
selection. 
\begin{example}\label{PRM not complete}
Consider the clauses
  \begin{align*}
    \text{Theory} &\left\{
    \begin{array}{l}
     \underline P\vee Q \\
      \underline{\neg P}\vee Q 
    \end{array}\right.
\\
  \text{Other clauses}&~\{~\neg Q
  \end{align*}
where the selected literals are underlined. We cannot derive the empty
clause if we restrict the application of the resolution rule to
clauses such that at most one of them is the theory and
the resolved literal in a theory clause is selected. 
However, the theory is consistent, so the
Set-of-support restriction alone is complete; and literals are selected 
according to the order
$Q\prec P$, so the Ordered resolution with selection alone also is complete.
\end{example}

\section{Polarized resolution modulo}

To prove the completeness of the combination of the Set-of-support method and
the Ordered resolution with selection method, we therefore need a
stronger condition than the consistency of the theory and the use of
an order to define selected literals. As we shall see, this condition
is exactly the cut elimination property for the polarized sequent calculus
modulo some rewrite rules corresponding to the theory clauses.
 
Indeed, the recently introduced {\em Polarized resolution modulo} method
\cite{polar}, combines these two restrictions.  In this method, we first 
identify
a subset of the set of clauses to be refuted. This set is called {\em 
the theory} and its elements 
{\em one-way clauses}. Then, in each of 
these clause, we identify 
a {\em selected literal}, and we impose the following restrictions:
\begin{itemize}
\item the resolution rule may be applied to a pair of clauses if 
at most one clause is a one-way clause, but it is forbidden when
both are, 

\item the resolution rule can be applied to a pair of clauses containing 
a one-way clause, if the resolved literal in the one-way clause is
the selected one, but not otherwise.
\end{itemize}
A last feature of Polarized resolution modulo is that unification is 
replaced by equational unification, but we shall not use this here.

\begin{example}\label{aaa}
Consider an arbitrary set of clauses containing the clauses
$$\underline{P}, Q$$
$$\underline{P}, \neg Q$$
then taking all the clauses of this subset to be one-way clauses and 
selecting the underlined literals is a complete restriction of resolution. 
\end{example}

\begin{example}\label{bbb}
Consider an arbitrary set of clauses containing the clauses
$$\underline{\neg \varepsilon(x~\dotvee~y)}, \varepsilon(x), \varepsilon(y)$$
$$\underline{\varepsilon(x~\dotvee~y)}, \neg \varepsilon(x)$$
$$\underline{\varepsilon(x~\dotvee~y)}, \neg \varepsilon(y)$$
$$\underline{\neg \varepsilon(\dotneg~x)}, \neg \varepsilon(x)$$
$$\underline{\varepsilon(\dotneg~x)}, \varepsilon(x)$$
$$\underline{\neg \varepsilon(\dotfa_T~x)}, \varepsilon(x~y)$$
$$\underline{\varepsilon(\dotfa_T~x)}, \neg \varepsilon(x~H_T(x))$$
$$\underline{\neg \varepsilon(\nulll~(S~x))}$$
$$\underline{\varepsilon(\nulll~0)}$$
and clauses containing no occurrences of the symbols $H_T$.
Then, taking all these clauses of this subset to be one-way clauses and 
selecting the underlined literals is a complete restriction of resolution.

Replacing unification with equational unification makes this method complete 
for Simple Type Theory \cite{polar}.
\end{example}

\begin{figure}
\noindent\framebox{\parbox{\textwidth
}{
$$
\hspace*{-4cm}
\begin{array}{c}
\irule{}
      {A \vdash B}
      {\mbox{axiom if $A \lra_-^* P, B \lra_+^* P$ and $P$ atomic}}
\vspace{1.5mm}\\
\irule{\Gamma, B \vdash \Delta ~~~ \Gamma \vdash C, \Delta}
      {\Gamma \vdash \Delta}
      {\mbox{cut if $A \lra_-^* B, A \lra_+^* C$}}
\vspace{1.5mm}\\
\irule{\Gamma, B, C \vdash \Delta}
      {\Gamma, A \vdash \Delta}
      {\mbox{contr-left if $A \lra_-^* B, A \lra_-^* C$}}
\vspace{1.5mm}\\
\irule{\Gamma \vdash B,C,\Delta}
      {\Gamma \vdash A,\Delta}
      {\mbox{contr-right if $A \lra_+^* B, A \lra_+^* C$}}
\vspace{1.5mm}\\
\irule{\Gamma \vdash \Delta}
      {\Gamma, A \vdash \Delta}
      {\mbox{weak-left}}
\vspace{1.5mm}\\
\irule{\Gamma \vdash\Delta}
      {\Gamma \vdash A,\Delta}
      {\mbox{weak-right}}
\vspace{1.5mm}\\
\irule{}
      {\Gamma, A \vdash \Delta}
      {\mbox{$\bot$-left if $A \lra_-^* \bot$}}
\vspace{1.5mm}\\
\irule{\Gamma \vdash B, \Delta}
      {\Gamma, A \vdash  \Delta}
      {\mbox{$\neg$-left if $A \lra_-^* \neg B$}}
\vspace{1.5mm}\\
\irule{\Gamma, B \vdash \Delta}
      {\Gamma \vdash  A, \Delta}
      {\mbox{$\neg$-right if $A \lra_+^* \neg B$}}
\vspace{1.5mm}\\
\irule{\Gamma, B \vdash \Delta ~~~ \Gamma, C \vdash \Delta}
      {\Gamma, A \vdash  \Delta}
      {\mbox{$\vee$-left if $A \lra_-^* (B \vee C)$}}
\vspace{1.5mm}\\
\irule{\Gamma \vdash B, C, \Delta}
      {\Gamma \vdash A, \Delta}
      {\mbox{$\vee$-right if $A \lra_+^* (B \vee C)$}}
\vspace{1.5mm}\\
\irule{\Gamma, C \vdash \Delta}
      {\Gamma, A \vdash  \Delta}
      {\mbox{$\langle x, B, t \rangle$ $\fa$-left if $A \lra_-^* \fa x~B$, $(t/x)B \lra_-^* C$}}
\vspace{1.5mm}\\
\irule{\Gamma \vdash B, \Delta}
      {\Gamma \vdash  A, \Delta}
      {\mbox{$\langle x, B \rangle$ $\fa$-right 
   if $A \lra_+^* \fa x~B$, $x \not\in FV(\Gamma \Delta)$}}
\end{array}
$$
\caption{Polarized sequent calculus modulo\label{sequent}}
}}
\end{figure}

As we saw in Example~\ref{PRM not complete}, this method is not always
complete even if the theory is consistent and if the
selected literals are maximal for some order on atoms.
But, we have proved in \cite{polar} that the completeness of this method is 
equivalent to cut elimination for the polarized sequent
calculus modulo 
the rewrite system associated to the set of one-way
clauses, where the rules of the polarized sequent calculus modulo are 
given in Fig. \ref{sequent} and the relation between clauses with 
selected literals and polarized rewrite rules is defined as follows.

\begin{definition}
Let $T$ be a set of clauses, such that in each clauses, a
literal is selected. The rewrite system associated with $T$ is
defined by:

To each selected literal $L$ in a clause $L, C_1, ..., C_p$
corresponds a 
rewrite
rule
\begin{itemize}
\item if $L$ is a negative literal $\neg P$, the rule  $P\lra_- \forall 
x_1 ... \forall x_n (C_1 \vee ... \vee C_p)$
\item if $L$ is a positive literal $P$, the rule $P\lra_+
\neg\forall x_1 ... \forall x_n (C_1 \vee ... \vee C_p)$
\end{itemize}
where $x_1$, ..., $x_n$ are the variables free in $C$ but not in $P$. 
\end{definition}

\begin{theorem}\label{PRM complete if and only if admits cut}
Let $T$ be a set of one-way clauses and $\mathcal R$ be the rewrite
system associated with $T$.  Polarized resolution modulo with the set
of one-way clauses $T$ is complete if and only if the polarized
sequent calculus 
modulo $\mathcal R$ admits the cut rule.
\end{theorem}
\begin{proof}
Corollary of Theorem~1 of \cite{polar}.
\end{proof}

For instance, the polarized sequent calculus modulo the
rewrite system
$$P \lra_+ \neg Q$$
$$P \lra_+ \neg \neg Q$$ has the cut elimination property
(Proposition~7 of~\cite{dowek:theory} proves that this property holds
whenever the left hand sides of the positive and of the negative rules
are disjoint sets) hence the completeness 
in Example \ref{aaa}. 

In the same way, the polarized sequent
calculus modulo the rewrite system
$$\begin{array}{rcl@{~}@{~}@{~}rcl}
\varepsilon(x~\dotvee~y) &\lra_-& (\varepsilon(x) \vee \varepsilon(y))&
\varepsilon(x~\dotvee~y) &\lra_+& \neg \neg \varepsilon(x)\\
&&& \varepsilon(x~\dotvee~y) &\lra_+& \neg \neg \varepsilon(y)\\
\varepsilon(\dotneg~x) &\lra_-& \neg \varepsilon(x) &
\varepsilon(\dotneg~x) &\lra_+& \neg \varepsilon(x)\\
\varepsilon(\dotfa_T~x) &\lra_-& \fa y~\varepsilon(x~y) &
\varepsilon(\dotfa_T~x) &\lra_+& \neg \neg \varepsilon(x~H_T(x))\\
&&&\varepsilon(\nulll~0) &\lra_+& \neg \bot\\
\varepsilon(\nulll~(S~x)) &\lra_-& \bot\\
\end{array}$$
has the cut elimination property for all sequents containing no occurrences
of the symbols $H_T$, hence the completeness in Example~\ref{bbb}. 

\section{Comparing Polarized resolution modulo with Set-of-support resolution}

In contrast, the polarized sequent calculus modulo the rewrite system 
$$P \lra_- Q$$
$$P \lra_+ \neg Q$$ does not have the cut elimination property. Indeed,
the proposition $Q$ has a proof by cutting through $P$, but no cut
free proof. This explains the incompleteness in Example~\ref{PRM not
complete}.

The next example shows that Polarized resolution modulo can fail in finite 
time for some input when Set-of-support resolution loops.
\begin{example}
  Consider the one-way clause $$\underline{P(f(x))},\neg P(x)$$ and
  another clause $$P(a)$$ for some constant $a$. Both Polarized
  resolution modulo and Set-of-support resolution are
  complete. However, Polarized resolution modulo fails in finite time
  to derive the empty clause (the two clauses cannot be resolved due
  to the condition on the selected literal), whereas Set-of-support
  resolution loops, deriving the clauses
  $P(\underbrace{f(\ldots(f(}_{i\text{ times}}a)\ldots))$ for all $i>0$.
\end{example}

In this example, we see the importance of having a selection
function. However, having such a selection function is not enough, 
as we will see in Section~\ref{PRM vs ORS}.

\section{Comparing Polarized resolution modulo with Ordered resolution with 
selection}\label{PRM vs ORS}

{\em Ordered resolution with selection} \cite{BachmairGanzinger} is a
proof search method pa\-ram\-e\-trized by a computable selection
function $S$, that associates to each clause a set of negative
literals of this clause and a decidable order relation $\prec$ on
atoms that is stable by substitution and total on ground atoms.  Each
pair $\langle S, {\prec} \rangle$ defines a different proof-search method. Thus
Ordered resolution with selection is a family of proof search methods rather
than a single method.

Many known restrictions of resolution are instances of Ordered resolution 
with selection for an appropriate selection function and order.  
Thus, we may wonder if Polarized resolution modulo is, in the same way 
an instance of Ordered resolution with selection. We prove now that this is 
not the case.

We prove more generally, that if $m$ and $m'$ are two proof-search methods 
and $T$ is 
a theory such that the completeness of $m$ can be proved in $T$ and $m'$ 
fails in finite time attempting to prove a contradiction in $T$, then $m$ 
and $m'$ are different methods: they are separated by the theory $T$. 

\subsection{Separation of proof-search methods}

Consider a decidable set of axioms $T$, that is an $\omega$-consistent 
extension 
of arithmetic. We can express in the language of $T$, a proposition 
$Bew$ with two free variables, such that if $U$ is a decidable set of axioms 
(i.e. the index of a total computable function characterizing these axioms) and $A$ is
a proposition (i.e. the index of a proposition) then the sequent
$T \vdash Bew(U, A)$ is provable in predicate logic if and only if 
the sequent $U \vdash A$ is.

Consider a proof-search method $m$. We can build, in the language of
$T$, a proposition $M$ with two free variables, such that if $U$ is a
decidable set of axioms and $A$ is a proposition 
then 
\begin{itemize}
\item if the method 
$m$ applied to the theory
$U$ and to the proposition $A$ succeeds then 
the sequent $T \vdash M(U, A)$ is
provable, 
\item if the method 
$m$ applied to the theory
$U$ and to the proposition $A$ does not succeed (i.e. fails in finite
time or loops) then 
the sequent $T \vdash M(U, A)$ is not provable, 
\item if the method $m$ applied to the theory
$U$ and to the proposition $A$ fails in finite time
then the sequent $T \vdash \neg M(U, A)$ is provable.
\end{itemize}

Assume, moreover that the completeness of the method $m$ is provable in $T$, 
i.e. that the sequent 
$$T \vdash \fa U \fa A~(Bew(U,A) \Rightarrow M(U,A))$$ 
is provable in predicate logic.
Then, by G\"odel second incompleteness theorem, the sequent
$$T \vdash \neg Bew(T,\bot)$$ 
is not provable in predicate logic,
hence the sequent 
$$T \vdash \neg M(T,\bot)$$ 
is not provable either.
Thus, if the completeness of $m$ can be proved in the theory $T$, then the 
method $m$ attempting to prove $\bot$ in the theory $T$ cannot fail in 
finite time: 
it must loop. 

If $m'$ is a proof-search method that fails in finite time
attempting to prove $\bot$ in the theory $T$, then $m$ and $m'$ are different.

\subsection{Translations}

This can be generalized in the following way. Assume that $\phi$ is a 
translation from theories to theories whose completeness can be proved in 
$T$, i.e. such that the sequent
$$T \vdash \fa U~(Bew(U,\bot) \Rightarrow Bew(\phi(U),\bot))$$ 
is provable. 
Assume, moreover that the completeness of the method $m$, is provable in $T$, 
i.e. that the sequent 
$$T \vdash \fa U \fa A~(Bew(U,A) \Rightarrow M(U,A))$$ 
is provable in predicate logic.
Then, by G\"odel second incompleteness theorem, the sequent
$$T \vdash \neg Bew(T,\bot)$$ is not provable in predicate logic, 
hence the sequent 
$$T \vdash \neg M(\phi(T),\bot)$$ 
is not provable either.
Thus, if the completeness of the translation $\phi$ and that of 
the method $m$ can be proved in $T$, 
then the method $m$ 
attempting to prove $\bot$ in the theory $\phi(T)$ cannot fail it finite time: 
it must loop.

If $m'$ is a proof-search method that fails in finite time
attempting to prove $\bot$ from $T$, then the application of 
$m'$ to a theory $U$ is not the application of $m$ to $\phi(U)$. 

\subsection{Application to Ordered resolution with selection}

Let ${\cal H}$ be the first-order presentation of Simple type theory
of Example \ref{bbb}. 
The completeness of Ordered resolution with selection can be proved in 
${\cal H}$ provided the stability and the totality of the order
are provable in ${\cal H}$. 

Attempting to prove $\bot$ in the theory ${\cal H}$ in Polarized resolution
modulo fails in finite time. 
Thus, Polarized resolution modulo is not an instance of Ordered 
resolution with selection for any selection function and order
relation whose stability and totality can be proved in ${\cal
H}$. Neither it is the application of Ordered resolution with selection
to a translation of its input, for a translation provably complete 
in ${\cal H}$.

\section{Conclusion}

Polarized resolution modulo is a combination of a restriction on the
choice of clauses and a restriction on the choice of literals in
resolution.  Combining these two restrictions makes the method so
restrictive that it is not always complete, even when the one-way
clauses form a consistent subset and the selection of literals is
based on an ordering on atoms.  But the completeness condition is
stronger, as completeness is equivalent to the cut elimination
property for the associated sequent calculus in Polarized deduction
modulo. Thus, unlike Ordered resolution with selection, the
completeness of all instances of the method cannot be proved in the
same theory.

The advantage of such methods whose completeness of all instances
cannot be proved in the same theory is that their logical strength is
not limited by the logical strength of the theory in which 
the completeness of the method can be proven.

\end{document}